\documentclass{amsart}
\pdfoutput=1
\usepackage{mathrsfs}
\usepackage[utf8]{inputenc}

\usepackage{xspace}
\usepackage{hyperref}
\usepackage{tikz}
\usepackage{blkarray} 
\usepackage{multirow}
\usepackage{enumerate}
\usepackage{a4wide}
\usepackage[all=normal]{savetrees}

\hypersetup{
pdftitle={Tree-depth and vertex-minors}, 
pdfauthor={Petr Hlineny, O-joung Kwon, jan Obdrzalek and Sebastian Ordyniak}
}
\newcommand\abs[1]{\lvert #1\rvert}
\newtheorem{THM}{Theorem}[section]
\newtheorem{DEF}[THM]{Definition}
\newtheorem{LEM}[THM]{Lemma}
\newtheorem{COR}[THM]{Corollary}
\newtheorem{PROP}[THM]{Proposition}
\newtheorem{CONJ}[THM]{Conjecture}

\newcommand\cl{\operatorname{Clos}}
\newcommand\td{\operatorname{td}}
\def\cf#1{{\mathscr#1}}
\def\mx#1{\mbox{\boldmath$#1$}}

\def\prebox#1{\mathop{\mbox{\rm #1}}}

\newcommand\modop{\operatorname{mod}}

\def\TM#1#2{{\mathcal{T\!M}}_{#2}(#1)}
\def\SC#1{{\mathcal{SC}}(#1)}
\def\BSC#1{{\mathcal{BSC}}(#1)}

\newcommand{\FO}{\ensuremath{\mathrm{FO}}\xspace}
\newcommand{\MSO}{\ensuremath{\mathrm{MSO}}\xspace}
\newcommand{\CMSO}{\ensuremath{\mathrm{CMSO}}\xspace}
\newcommand{\MSOi}{\ensuremath{\mathrm{MSO}_1}\xspace}
\newcommand{\MSOii}{\ensuremath{\mathrm{MSO}_2}\xspace}
\newcommand{\CMSOi}{\ensuremath{\mathrm{CMSO}_1}\xspace}
\newcommand{\PMSOi}{\ensuremath{\mathrm{C_2MSO}_1}\xspace}

\theoremstyle{definition}

\newcommand{\shortversion}[1]{}
\newcommand{\longversion}[1]{#1}

\begin{document}
\title[Tree-depth and vertex-minors]{Tree-depth and Vertex-minors}
\author{Petr Hliněný}
\address[P.~Hliněný, J.~Obdržálek, S.~Ordyniak]{Faculty of Informatics, Masaryk University,
	Botanick\'a 68a, Brno, Czech Republic}
\email{\{hlineny,obdrzalek\}@fi.muni.cz, sordyniak@gmail.com}
\author{O-joung Kwon}
\address[O.~Kwon]{Department of Mathematical Sciences, KAIST, 291 Daehak-ro
  Yuseong-gu Daejeon, 305-701 South Korea} 
\email{ojoung@kaist.ac.kr}
\author{Jan Obdržálek}
\author{Sebastian Ordyniak}
\thanks{P.~Hliněný and J.~Obdržálek have been supported by 
	the Czech Science Foundation, project no.~14-03501S.
	S.~Ordyniak have been supported by the European Social Fund 
	and the state budget of the Czech Republic
	under project CZ.1.07/2.3.00/30.0009 (POSTDOC I)}
\thanks{O.~Kwon have been supported by Basic Science Research
  Program through the National Research Foundation of Korea (NRF)
  funded by the Ministry of Education, Science and Technology
  (2011-0011653).}
\date{}

\begin{abstract}
In a recent paper~\cite{Kwon2013}, Kwon and Oum claim that every graph of bounded
rank-width is a pivot-minor of a graph of bounded tree-width
(while the converse has been known true already before).
We study the analogous questions for ``depth'' parameters of graphs,
namely for the tree-depth and related new shrub-depth.
We show that shrub-depth is monotone under taking vertex-minors,
and that every graph class of bounded shrub-depth can be obtained via
vertex-minors of graphs of bounded tree-depth.
We also consider the same questions for bipartite graphs and pivot-minors.
\end{abstract}

\maketitle

\section{Introduction} \label{sec:intro}

Various notions of graph containment relations (e.g. graph minors) play an important part in structural graph
theory. Recall that a graph $H$ is a minor of a graph $G$ if $H$ can be
obtained from $G$ by a sequence of edge contractions, edge deletions and
vertex deletions. In their seminal series of papers, Robertson and Seymour
introduced the notion of tree-width and showed the following: The tree-width of
a minor of $G$ is at most the tree-width of $G$ and, moreover, for each $k$
there is a finite list of graphs such that a graph $G$ has tree-width at
most $k$ if, and only if, no graph in the list is isomorphic to a minor of
$G$. This, among other things, implies the existence of a polynomial-time
algorithm to check that the tree-width of a graph is at most $k$.

There have been numerous attempts to extend this result to (or find a similar
result for) ``width'' measures other than tree-width . The
most natural candidate is clique-width, a measure generalising tree-width defined by Courcelle and
Olariu~\cite{CO00}. However, the quest to prove a similar result for this
measure has been so far unsuccessful. For one, taking the graph minor relation is
clearly not sufficient as every graph on $n>1$ vertices is a minor of the
complete graph $K_n$, clique-width of which is 2.

However Oum~\cite{Oum05} succeeded in finding the appropriate containment
relation -- called \emph{vertex-minor} -- for the notion of rank-width,
which is closely related to clique-width. (More precisely, if the clique-width of
a graph is $k$, then its rank-width is between $\log_2(k+1)-1$ and $k$.)
Vertex-minors are based on the operation of local complementation: taking a
vertex $v$ of a graph $G$ we replace the subgraph induced on the neighbours
of $v$ by its edge-complement, and denote the resulting graph by $G*v$. We then
say that a graph $H$ is a vertex-minor of $G$ if $H$ can be obtained from
$G$ by a sequence of local complementations and vertex
deletions. In~\cite{Oum05} it was shown that if $H$ is a vertex-minor of
$G$, then its rank-width is at most the rank-width of $G$.

Another graph containment relation, the \emph{pivot-minor}, also
defined in~\cite{Oum05}, is closely related to vertex-minor. Pivot-minors are
based on the operation of edge-pivoting: for an edge $e=\{u,v\}$ of a graph $G$ we
perform the operation $G*u*v*u$.  Then a graph $H$ is a pivot-minor of $G$ if
it can be obtained from $G$ by a sequence of edge-pivotings and vertex
deletions. It follows from the definition that  every pivot-minor is
also a vertex-minor.

This brings an interesting question: What is the exact relationship between
various width measures with respect to these new graph containment
relations? Recently, it was shown that every graph of rank-width $k$ is a
pivot-minor of a graph of tree-width at most $2k$~\cite{Kwon2013}. In this
paper we investigate the existence of similar relationships for two ``shallow''
graph width measures: tree-depth and shrub-depth.

\emph{Tree-depth}~\cite{no06} is a graph invariant which
intuitively measures how far is a graph from being a star. Graphs of bounded
tree-depth are sparse and play a central role in the theory of graph
classes of bounded expansion. \emph{Shrub-depth}~\cite{GHNOOR12} is a very
recent graph invariant, which was designed to fit into the gap between
tree-depth and clique-width. (If we consider tree-depth to be the
``shallow'' counterpart of tree-width, then shrub-depth can be thought
of as a ``shallow'' counterpart of clique-width.)

Our results can be summarised as follows. We start by showing that
shrub-depth is monotone under taking vertex-minors
(Corollary~\ref{cor:shrubd-monotone}). Next we prove that every graph class
of bounded shrub-depth can be obtained via vertex-minors of graphs of
bounded tree-depth (Theorem~\ref{thm:tdtoshrubd}). Note that, unlike for
rank-width and tree-width, restricting ourselves to pivot-minors is not
sufficient. Indeed, this is because, as we prove in~Proposition~\ref{prop:noclique}, graphs of bounded
tree-depth cannot contain arbitrarily large cliques as pivot-minors.
Interestingly, we are however able to show the same result for pivot-minors
if we restrict ourselves to bipartite graphs, which were, in a similar connection, investigated already
in \cite{Kwon2013}. 
\longversion{
In particular, our main result of the last section is that for any
class of bounded shrub-depth there exists an integer $d$ such that any
bipartite graph in the class is a pivot-minor of a graph of tree-depth $d$.
}
\shortversion{
In particular, we show that for any
class of bounded shrub-depth there exists an integer $d$ such that any
bipartite graph in the class is a pivot-minor of a graph of tree-depth $d$.
}

\longversion{
\section{Preliminaries} \label{sec:prelim}

	In this paper, all graphs are finite, undirected and simple.
	A {\em tree} is a connected graph with no cycles, and it is 
	{\em rooted} if some vertex is designated as the root.
	A leaf of a rooted tree is a vertex other than the root having just
	one neighbour.
	The height of a rooted tree is the maximum length of a path starting
	in the root (and hence ending in a leaf).
	Let $G$ be a graph. 
	We denote $V(G)$ as the vertex set of $G$ and
	$E(G)$ as the edge set of $G$.
	For $v\in V(G)$, let $N_G(v)$ be the set of the neighbours of $v$ in $G$.

	We sometimes deal with {\em labelled graphs} $G$, which means that
	every vertex of $G$ is assigned a subset (possibly empty) of a given
	finite label set.
	A graph is {\em$m$-coloured} if every vertex is assigned exactly one
	of given $m$ labels (this notion has no relation to ordinary graph colouring).

We now briefly introduce the \emph{monadic second order logic} (\MSO) over graphs
and the concept of \FO (\MSO) graph interpretation. \MSO is the extension of
first-order logic (\FO) by quantification over sets, and comes in two
flavours, \MSOi and \MSOii, differing by the objects we are allowed to
quantify over: 
\begin{DEF}[\MSOi logic of graphs]
  The language of \MSOi consists of expressions built from the following
  elements:
  \begin{itemize}
  \item variables $x,y,\ldots$ for vertices, and $X,Y$ for sets of vertices,
  \item the predicates $x\in X$ and $\prebox{edge}(x,y)$ with the standard
   meaning,
  \item equality for variables, quantifiers $\forall$ and $\exists$ ranging over
   vertices and vertex sets, and the standard Boolean connectives.
  \end{itemize}
\end{DEF}
\MSOi logic can be used to express many interesting graph
properties, such as 3-colour\-ability. We also mention \MSOii logic, which additionally includes quantification over edge sets and can express properties which are not \MSOi definable (e.g. Hamiltonicity).
The large expressive power of
both \MSOi and \MSOii makes them a very popular choice when
formulating algorithmic metatheorems (e.g., for graphs of bounded clique-width or
tree-width, respectively). 

The logic we will be mostly concerned with is an extension of \MSOi called
\emph{Counting monadic second-order logic} (\CMSOi). In addition to the
\MSOi syntax \CMSOi allows the use of
predicates $\modop_{a,b}(X)$, where $X$ is a set variable. The semantics of
the predicate $\modop_{a,b}(X)$ is that the set $X$ has $a$ modulo $b$
elements. We use \PMSOi to denote the parity counting fragment of \CMSOi,
i.e. the fragment where the predicates $\modop_{a,b}(X)$ are restricted to $b=2$.

A useful tool when solving the model checking problem on a class of
structures is the ability to ``efficiently translate'' an instance of the problem to a
different class of structures, for which we already have an efficient model checking
algorithm. To this end we introduce simple FO/\MSOi graph interpretation, which is
an instance of the general concept of interpretability of logic
theories~\cite{Rab64} restricted to simple graphs with vertices
represented by singletons.

\begin{DEF}
\label{def:interpretation}
  A {\em FO (\MSOi) graph interpretation} is a pair $I=(\nu,\mu)$ of
  {\em FO} (\MSOi) formulae (with $1$ and $2$ free variables respectively) in the
  language of graphs, where $\mu$ is symmetric (i.e.
  $G\models\mu(x,y)\leftrightarrow\mu(y,x)$ in every graph $G$).  To each
  graph $G$ it associates a graph $G^I$,
  which is defined as follows:
\begin{itemize}
  \item The vertex set of $G^I$ is the set of all vertices $v$ of $G$ such
  that $G\models \nu(v)$;
  \item The edge set of $G^I$ is the set of all the pairs $\{u,v\}$ of
  vertices of $G$ such that $G\models \nu(u)\wedge\nu(v)\wedge\mu(u,v)$.
\end{itemize}
\end{DEF}
This definition naturally extends to the case of vertex-labelled graphs (using a
finite set of labels, sometimes called colours) 
by introducing finitely many unary relations in the language to encode the labelling.

For example, a complete graph can be interpreted in any graph (with the same
number of vertices) by letting $\nu\equiv\mu\equiv true$,
and the complement of a graph has an interpretation using
$\mu(x,y)\equiv\neg\prebox{edge}(x,y)$.

\subsection*{Vertex-minors and Pivot-minors}
For $v\in V(G)$, the \emph{local complementation} at a vertex $v$ of $G$ is
the operation which complements the adjacency between every pair of two
vertices in $N_G(v)$.  The resulting graph is denoted by $G*v$. We say that
two graphs are \emph{locally equivalent} if one can be obtained from the
other by a sequence of local complementations.  For an edge
$uv\in E(G)$, \emph{pivoting} an edge $uv$ of $G$ is defined as $G\wedge
uv=G*u*v*u=G*v*u*v$.  A graph $H$ is a \emph{vertex-minor} of $G$ if $H$ is
obtained from $G$ by applying a sequence of local complementations and
deletions of vertices.  A graph $H$ is a pivot-minor of $G$ if $H$ is
obtained from $G$ by applying a sequence of pivoting edges and deletions of
vertices.  From the definition of pivoting every pivot-minor of a graph is also its
vertex-minor.

 Pivot-minors of graphs are closely related to a matrix operation called
 pivoting. To give the exact relationship
 (Proposition~\ref{prop:pivots-minors-matrices}) we will need to introduce
 some matrix concepts.

 \subsection*{Pivoting on a Matrix}

	For two sets $A$ and $B$, we denote by $A\Delta B=(A\setminus B)\cup
        (B\setminus A)$ its symmetric difference.
	Let $\mx M$ be a $S\times T$ matrix. 
	For $A\subseteq S$ and $B\subseteq T$,
	we denote the $A\times B$ submatrix of $\mx M$ as 
	$\mx M[A,B]=(m_{i,j})_{i\in A, j\in B}$.
	If $A=B$, then $\mx M[A]=\mx M[A,A]$ and we call it a \emph{principal submatrix} of
	$\mx M$.
	If $a\in S$ and $b\in T$, then 
	we denote $\mx M_{a,b}=\mx M[\{a\}, \{b\}]$.
	The \emph{adjacency matrix} $\mx A(G)$ of $G$ is the $V(G)\times V(G)$ matrix such that 
	for $v,w\in V(G)$, $\mx A(G)_{v,w}=1$ if $v$ is adjacent to $w$ in $G$, and 
	$\mx A(G)_{v,w}=0$ otherwise.

	Let \[
\mx M=\bordermatrix{
	& S & X\setminus S\cr
	S & A & B\cr
	X\setminus S & C & D
}
\] be a $X\times X$ matrix over a field $F$. 

If $\mx A=\mx M[S]$ is non-singular, then we define \emph{pivoting} $S$ 
on the matrix $\mx M$ as \[ 
\mx M\ast S=
\bordermatrix{
	 &  S & X\setminus S\cr
 \hfill S & A^{-1} & A^{-1}B \cr
	X\setminus S & -CA^{-1}  & D-CA^{-1}B 
}.
\]
      It is sometimes called a \emph{principal pivot transformation}~\cite{Tsatsomeros2000}.
	The following theorem is useful when dealing with matrix pivoting.

\begin{THM}[Tucker \cite{Tucker60}]\label{thm:21}
Let $\mx M$ be a $X\times X$ matrix over a field.
If $\mx M[S]$ is a non-singular principal submatrix of $\mx M$, 
then for every $T\subseteq X$, 
$(\mx M\ast S)[T]$ is non-singular if and only if $\mx M[S\Delta T]$ is non-singular.
\end{THM}
\begin{proof}
  See Bouchet's proof in Geelen~\cite[Theorem 2.7]{Geelen1995}.
\end{proof}

\begin{THM}\label{thm:22}
Let $\mx M$ be a $X\times X$ matrix over a field.
If $\mx M[S]$ and $(\mx M*S)[T]$ are non-singular, 
then $(\mx M*S)*T=\mx M*(S\Delta T)$.
\end{THM}
\begin{proof}
See Geelen~\cite[Theorem 2.8]{Geelen1995}.
\end{proof}

We are now ready to state the relationship between pivot-minors and matrix
pivots. The proof of the following proposition uses Theorem~\ref{thm:21} and
Theorem~\ref{thm:22}, and we refer the reader to~\cite{Kwon2013} for
detailed explanation.

\begin{PROP}
  \label{prop:pivots-minors-matrices}
  Graph $H$ is a pivot-minor of $G$ if and only if $H$ is the graph whose
  adjacency matrix is $(\mx A(G)*X)[Y]$ where $X, Y\subseteq V(G)$ and $\mx
  A(G)[X]$ is non-singular.
\end{PROP}

\subsection*{Tree-depth}
	For a forest $T$, 
	the closure $\cl (T)$ of $T$ is the graph obtained from $T$ by making  every vertex adjacent to all of its ancestors.
	The \emph{tree-depth} of a graph $G$, denoted by $\td(G)$, is one more than the minimum height of a rooted forest $T$ such that $G\subseteq \cl(T)$.
}

\section{Shrub-depth and Vertex-minors}

In this section we show the first of our results -- that shrub-depth is
monotone under taking vertex-minors. The shrub-depth of a graph class is
defined by the following very special kind of a simple \FO interpretation:

\begin{DEF}[Tree-model \cite{GHNOOR12}]
\label{def:tree-model}
We say that a graph $G$ has a {\em tree-model of $m$ colours and depth $d$}
if there exists a rooted tree $T$ (of height $d$) such that: 
\begin{enumerate}[i.]
\item the set of leaves of $T$ is exactly $V(G)$,
\item the length of each root-to-leaf path in $T$ is exactly~$d$,
\item each leaf of $T$ is assigned one of $m$ colours
(\,i.e. $T$ is {\em $m$-coloured}), 
\item\label{it:tree-model-edge}
and the existence of an edge between $u,v\in V(G)$ depends solely
on the colours of $u,v$ and the distance between $u,v$ in $T$.
\end{enumerate}
The class of all graphs having such a tree-model is denoted by $\TM dm$.
\end{DEF}

For example, $K_n\in\TM11$ or $K_{n,n}\in\TM12$. We thus consider:

\begin{DEF}[Shrub-depth \cite{GHNOOR12}]
\label{def:shrub-depth}
A class of graphs $\cf S$ has {\em shrub-depth} $d$
if there exists $m$ such that $\cf S\subseteq\TM dm$,
while for all natural $m$ it is $\cf S\not\subseteq\TM{d-1}m$.
\end{DEF}

It is easy to see that each class $\TM dm$ is closed under complements and
induced subgraphs, but neither under disjoint unions, nor under
subgraphs. However, the class $\TM dm$ is not closed under local
complementations. On the other hand, 
to prove that shrub-depth is closed under vertex-minors it is sufficient to show that for
each $m$ there exists $m'$ such that all graphs locally equivalent to those
in $\TM dm$ belong to $\TM d{m'}$. As shrub-depth does not depend on $m$,
this will be our proof strategy.
\longversion{
Note that Definition~\ref{def:shrub-depth} is asymptotic as it makes sense 
only for infinite graph classes;
the shrub-depth of a single finite graph is always at most one.
For instance, the class of all cliques has shrub-depth $1$.
More interestingly, graph classes of certain shrub-depth are characterisable
exactly as those having simple \CMSOi interpretations in
the classes of rooted labelled trees of fixed height:

\begin{THM}[\cite{GHNOOR12,GH14}]
\label{thm:shrub-depth-interpretability}
A class $\cf S$ of graphs has a simple \CMSOi interpretation in the class of all
finite rooted labelled trees of height $\leq d$
if, and only if, $\cf S$~has shrub-depth at most~$d$.
\end{THM}
\begin{proof}[Proof sketch]
In \cite{GHNOOR12} this statement occurs with a little shift---involving
\MSOi logic instead of \CMSOi.
However, since the proof in \cite{GHNOOR12} builds everything on one technical
claim (kernelization of \MSO on trees of bounded height)
which has been subsequently extended to \CMSO in \cite[Section~3.2]{GH14},
the full statement follows as well.
\end{proof}
Note that the above theorem implies that
any class of graphs of bounded shrubdepth is closed under simple \CMSOi
interpretations, i.e., the class of graphs obtained via a simple
\CMSOi interpretation on a class of graphs of bounded shrub-depth
has itself bounded shrub-depth. This is one of the two essential ingredients we need to prove that
shrub-depth is closed under vertex-minors. The other ingredient is the following technical claim:

\begin{LEM}[Courcelle and Oum \cite{CO07}]
\label{lem:interpret-localeq}
For a graph $G$, let $\cf L(G)$ denote the set of graphs which are locally
equivalent to $G$.
Then there exists a simple \PMSOi interpretation such that each such
$\cf L(G)$ is interpreted in vertex-labellings of~$G$.
\end{LEM}
\begin{proof}[Proof sketch]
Again, \cite[Corollary~6.4]{CO07} states nearly
the same what we claim here.
The only trouble is that \cite{CO07} speaks about
more general so-called transductions.
Here we briefly survey that the transduction constructed in
\cite[Corollary~6.4]{CO07} is really a simple \PMSOi
interpretation (we have to stay on an informal level since a formal
introduction to all necessary concepts would take up several pages):
\begin{enumerate}[i.]
\item 
In \cite{CO07} local complementations of a
graph $G$ are treated via a so called isotropic system $S=S(G)$.
It is, briefly, a set of $V(G)$-indexed three-valued vectors,
and so $S$ can be described on the ground set $V(G)$ by a collection
of triples of disjoint sets.
This representation is definable in \PMSOi
\cite[Proposition~6.2]{CO07}.
\item 
The set of graphs locally equivalent to $G$ then corresponds to the
set of isotropic systems strongly isomorphic to $S$.
A strong isomorphism of isotropic systems on the ground set $V(G)$
is expressed in \MSOi with respect to a suitable $6$-partition of $V(G)$ by
\cite[Proposition~6.1]{CO07}.
\item 
Finally, a graph $H$ is locally equivalent to $G$ if and only if $H$ is the fundamental
graph of some (not unique) $S'\simeq S$ with respect to a special vector
of~$S'$, which again has a \PMSOi expression with respect to a triple of
subsets of $V(G)$ describing the vector (as in point i.)
by \cite[Proposition~6.3]{CO07}.
\end{enumerate}
Note that all the aforementioned \PMSOi expressions are on the same ground
set $V(G)$.
In the desired interpretation $I$ we treat the nine parameter sets of (ii.) and
(iii.) as a vertex-labelling of $G$, which consequently can interpret any $H$
locally equivalent to $G$ using \PMSOi.
\end{proof}
}

\begin{THM}
\label{thm:shrub-local-equal}
For a graph class $\cf C$, let $\cf L(\cf C)$ denote the class of graphs which 
are locally equivalent to a member of $\cf C$.
Then the shrub-depth of $\cf L(\cf C)$ is equal to the shrub-depth of $\cf C$.
\end{THM}
\longversion{\begin{proof}
Let $d$ be the least integer such that, for some $m$ as in
Definition~\ref{def:shrub-depth}, it is $\cf C\subseteq\TM dm$.
Let $I$ denote an \FO interpretation of $\cf C$ in
the class $\cf T_d$ of rooted labelled trees of height $d$
which naturally follows from Definition~\ref{def:tree-model},
and let $J$ be the simple \PMSOi interpretation from
Lemma~\ref{lem:interpret-localeq}.

For every $H\in\cf L(\cf C)$ there is a suitably labelled graph $G\in\cf C$ 
such that $H\simeq G^J$, and a tree $T\in\cf T_d$ such that $G\simeq T^I$.
As this $T$ can additionally inherit any suitable labelling of
$G$, we can claim $H\simeq (T^I)^J$.
Therefore, the composition $J\circ I$ is a \PMSOi interpretation of 
$\cf L(\cf C)$ in $\cf T_d$.
By Theorem~\ref{thm:shrub-depth-interpretability}, $\cf L(\cf C)$ is of
shrub-depth at most~$d$ and, at the same time, $\cf C\subseteq\cf L(\cf C)$.
\end{proof}
}

\begin{COR}
\label{cor:shrubd-monotone}
The shrub-depth parameter is monotone under taking vertex-minors over graph
classes.
\end{COR}
\longversion{
\begin{proof}
By the definition, a vertex-minor is obtained as an induced subgraph of a
locally equivalent graph.
Since taking induced subgraphs does not change a tree-model, the claim
follows from Theorem~\ref{thm:shrub-local-equal}.
\end{proof}
}

\section{From small Tree-depth to small SC-depth}

We have just seen that taking vertex-minors does not increase the shrub-depth of
a graph class.
It is thus interesting to ask whether, perhaps, every class of bounded
shrub-depth could be constructed by taking vertex-minors of some special
graph class.
This indeed turns out to be true in a very natural way---the special classes in
consideration are the graphs of bounded tree-depth.

Before proceeding we need to introduce another ``depth'' parameter
asymptotically related to shrub-depth which, unlike the former,
is defined for any single graph.
Let $G$ be a graph and let $X\subseteq V(G)$.
We denote by $\overline{G}^X$ the graph $G'$ with vertex set $V(G)$ where
$x\neq y$ are adjacent in $G'$ if either 
\longversion{
\begin{enumerate}[(i)]
\item $\{x,y\}\in E(G)$ and $\{x,y\}\not\subseteq X$, or
\item $\{x,y\}\not\in E(G)$ and $\{x,y\}\subseteq X$.
\end{enumerate}
}
\shortversion{
(i) $\{x,y\}\in E(G)$ and $\{x,y\}\not\subseteq X$, or
(ii) $\{x,y\}\not\in E(G)$ and $\{x,y\}\subseteq X$.
}
In other words, $\overline{G}^X$ is the graph obtained from $G$ 
by complementing the edges on~$X$.

\begin{DEF}[SC-depth \cite{GHNOOR12}]
\label{def:SC-depth}
We define inductively the class $\SC k$ as follows:
\begin{enumerate}[i.]
  \item let $\SC0=\{K_1\}$;
  \item if $G_1,\dots,G_p\in\SC k$ and $H= G_1\dot\cup\dots\dot\cup G_p$
  denotes the disjoint union of the $G_i$,
  then for every subset $X$ of vertices of $H$ we
  have $\overline{H\,}^X\in\SC{k+1}$.
\end{enumerate}
The {\em SC-depth} of $G$ is the minimum integer $k$ such that $G\in\SC k$.
\end{DEF}


\longversion{
\begin{PROP}[\cite{GHNOOR12}]
\label{thm:shrubsc}
The following are equivalent for any class of graphs $\cf G$:
\begin{itemize}
  \item there exist integers $d$, $m$ such that $\cf G\subseteq \TM dm$;
  \item there exists an integer $k$ such that $\cf G\subseteq \SC k$.
\end{itemize}
\end{PROP}
}
From Definition~\ref{def:SC-depth}, one can obtain the following claim:

\begin{LEM}\label{lem:vertex-minor-tree-depth}
	Let $k$ be a positive integer.
	If a graph $G$ has SC-depth at most $k$, then  
	$G$ is a vertex-minor of a graph of tree-depth at most $k+1$.
\end{LEM}
\longversion{
\begin{proof}
	For a graph $G$ of SC-depth $k$,
	we recursively construct a graph $U$ and a rooted forest $T$ such that
	\begin{enumerate}[i.]
	\item $G$ can be obtained from $U$ as a vertex-minor via applying local complementations
	only at the vertices in $V(U)\setminus V(G)$, and
	\item $U\subseteq \cl(T)$ and $T$ has depth $k$.
	\end{enumerate}
	If $k=0$, then it is clear by setting $G=U=T=K_1$. 
	We assume that $k\ge 1$.
			
	Since $G$ has SC-depth $k$, there exist a graph $H$ and $X\subseteq V(H)$ such that $G=\overline{H}^X$ 
	and $H$ is the disjoint union of $H_1, H_2, \ldots, H_m$ such that each $H_i$ has SC-depth $k-1$.
	By induction hypothesis, for each $1\le i\le m$, $H_i$ is a vertex-minor of a graph $U_i$ and $U_i\in \cl(T_i)$ where the
	height of $T_i$ is at most $k$.
	For each $1\le i\le m$, let $r_i$ be the root of $T_i$, and let $T$
	be the rooted forest obtained from the disjoint union of all $T_i$
	by adding a root $r$ which is adjacent to all $r_i$.
	Let $U$ be the graph obtained from the disjoint union of all $U_i$
	and $\{r\}$ by adding all edges from $r$ to $X$. 
	Validity of (ii.) is clear from the construction.
	
	Now we check the statement (i.). 
	By our construction of $U$, any local complementation in $U_i$ has no effect
	on $U_j$ for $j\not=i$, and local complementations at vertices
	in $V(U_i)\setminus V(H_i)$ do not change edges incident with $r$.
	Hence, by induction, we can obtain $H$ as a vertex-minor of $U$ and
	still have $r$ adjacent precisely to $X\subseteq V(H)$.
	We then apply the local complementation at $r\in V(U)\setminus
	V(H)$, and delete $V(U)\setminus V(G)$ to obtain $G$.
%
\end{proof}
}

\shortversion{
It has been shown in~\cite{GHNOOR12} that
a class of graphs $\cf G$ has SC-depth at most $k$ if and only if there
are integers $d$, $m$ such that $\cf G\subseteq \TM dm$. Using this characterisation
and the above lemma gives us the main conclusion of this section.
}

\longversion{
This, with Proposition~\ref{thm:shrubsc}, now immediately gives the main conclusion:
}
\begin{THM}
\label{thm:tdtoshrubd}
For any class $\cf S$ of bounded shrub-depth, there exists an integer $d$
such that every graph in $\cf S$ is a vertex-minor of a graph of tree-depth~$d$.
\qed\end{THM}

	Comparing Theorem~\ref{thm:tdtoshrubd} with \cite{Kwon2013} one may naturally ask whether,
	perhaps, weaker pivot-minors could be sufficient in
	Theorem~\ref{thm:tdtoshrubd}.
	Unfortunately, that is very false from the beginning.
	Note that all complete graphs have SC-depth $1$. 
	On the other hand, we will prove (Proposition~\ref{prop:noclique})
	that graphs of bounded tree-depth cannot contain
	arbitrarily large cliques as pivot-minors.
\longversion{
	We need the following technical lemmas.

\begin{LEM}\label{lem:avoidingedge}
	Let $G$ be a graph and $X\subseteq V(G)$ such that $\mx A(G)[X]$ is non-singular
	and $\abs{X}\ge 3$.
	If $u\in X$, then there exist $v, w\in X\setminus \{u\}$
	such that $vw\in E(G)$.
\end{LEM}
\begin{proof}
	Let $u\in X$. Suppose that for every pair of distinct vertices 
	$v, w\in X\setminus \{u\}$, $vw\notin E(G)$.
	That means $G[X]$ is isomorphic to a star with the centre $u$.
	However, the matrix $\mx A(G)[X]$ is clearly singular, 
and it contradicts to the assumption.
\end{proof}

\begin{LEM}\label{lem:reorder}
	Let $G$ be a graph and let $X\subseteq V(G)$ such that
	$X\neq \emptyset$ and $\mx A(G)[X]$ is non-singular.
	Let $s\in X$.
	Then $G$ has a sequence of pairs of vertices $\{x_1,y_1\},
	\{x_2,y_2\},$ $\ldots, \{x_m, y_m\}$ such that 
	\begin{enumerate}[\quad a)]
	\item $\mx A(G)*X=\mx A(G\wedge x_1y_1\wedge x_2y_2 \cdots \wedge x_my_m)$,
	\item $(\{x_i, y_i\}: {1\le i\le m})$ is a partition of~$X$
		(in particular, $\abs X$ is even), and
	\item $s\in \{x_m,y_m\}$.
	\end{enumerate}
%
\end{LEM}
\begin{proof}
	We prove the theorem by induction on $\abs{X}\geq1$.
	If $\abs X=1$, then $\mx A(G)[X]$ cannot be non-singular, as we have no loops in~$G$.
	If $X=\{x_1, x_2\}$, 
	then $x_1, x_2$ must form an edge of $G$ since, again, 
	$\mx A(G)[X]$ is non-singular. 
	Since $\mx A(G)*\{x_1, x_2\}=\mx A(G\wedge x_1x_2)$, 
	and either $s=x_1$ or $s=x_2$, we conclude the claim.

	For an inductive step, we assume that $\abs{X}\ge3$.
	Since $\mx A(G)[X]$ is non-singular, by Lemma~\ref{lem:avoidingedge}, 
	there exist two vertices $x_1, y_1\in X\setminus \{s\}$ such that
	$x_1y_1\in E(G)$.
	Also, by Theorem~\ref{thm:21}, 
	$\mx A(G\wedge x_1y_1)[X\setminus \{x_1, y_1\}]$ is non-singular. 
	By Theorem~\ref{thm:22}, we have
	\begin{align*}
	\mx A(G)*X
	&=\mx A(G)*(\{x_1,y_1\}\Delta (X\setminus \{x_1, y_1\}) \\
	&=(\mx A(G)*\{x_1,y_1\})*(X\setminus \{x_1, y_1\}) \\
	&=\mx A(G\wedge x_1y_1)*(X\setminus \{x_1, y_1\}).
	 \end{align*}
	
	 Since $s\in X\setminus \{x_1, y_1\}\not=\emptyset$, by the induction hypothesis, 
	 $G\wedge x_1y_1$ has a sequence of pairs of vertices $\{x_2,y_2\}, \ldots, \{x_m,y_m\}$ such that 
	 \begin{enumerate}[a)]
	\item $\mx A(G\wedge x_1y_1)*(X\setminus \{x_1, y_1\})=\mx A((G\wedge x_1y_1)\wedge x_2y_2 \cdots \wedge x_my_m)$,
	\item $(\{x_i, y_i\}: {2\le i\le m})$ is a partition of~$X\setminus\{x_1, y_1\}$, and
	\item $s\in \{x_m,y_m\}$.
	\end{enumerate}
        \sloppypar
	Thus, $\mx A(G)*X=\mx A(G\wedge x_1y_1\wedge x_2y_2 \cdots \wedge x_my_m)$ 
	and we can easily verify that
	$\{x_1, y_1\}, \{x_2, y_2\}, \ldots, \{x_m, y_m\}$ is the desired sequence.
\end{proof}

	Now we are ready to prove the promised negative proposition.

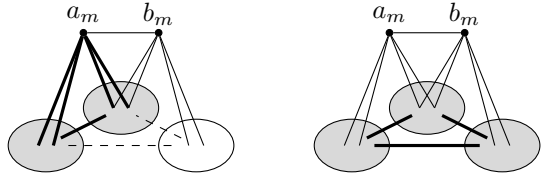
\begin{figure}[t]\centering
\tikzstyle{v}=[circle, draw, solid, fill=black, inner sep=0pt, minimum width=2.5pt]
\tikzset{photon/.style={decorate, decoration={snake}}}
\begin{tikzpicture}[scale=0.05]

\node [v] (a) at (10,50) {};
\node [v] (b) at (30,50) {};

\draw[fill=gray!30] (0,20) ellipse (10 and 7);
\draw (40,20) ellipse (10 and 7);
\draw[fill=gray!30] (20,30) ellipse (10 and 7);

\draw[very thick] (a) -- (-2, 20);
\draw[very thick] (a) -- (2, 20);
\draw (b) -- (42, 20);
\draw (b) -- (38, 20);

\draw[very thick] (a) -- (22, 30);
\draw[very thick] (a) -- (18, 30);
\draw (b) -- (22, 30);
\draw (b) -- (18, 30);

\draw (a) -- (b);

\draw[very thick] (4, 22) -- (16, 28);
\draw[dashed] (36, 22) -- (24, 28);
\draw[dashed] (6, 20) -- (34, 20);

\draw (10,55) node{$a_m$};
\draw (30,55) node{$b_m$};

\end{tikzpicture}\qquad\quad
\begin{tikzpicture}[scale=0.05]

\node [v] (a) at (10,50) {};
\node [v] (b) at (30,50) {};

\draw[fill=gray!30] (0,20) ellipse (10 and 7);
\draw[fill=gray!30] (40,20) ellipse (10 and 7);
\draw[fill=gray!30] (20,30) ellipse (10 and 7);

\draw (a) -- (-2, 20);
\draw (a) -- (2, 20);
\draw (b) -- (42, 20);
\draw (b) -- (38, 20);

\draw (a) -- (22, 30);
\draw (a) -- (18, 30);
\draw (b) -- (22, 30);
\draw (b) -- (18, 30);

\draw (a) -- (b);

\draw[very thick] (4, 22) -- (16, 28);
\draw[very thick] (36, 22) -- (24, 28);
\draw[very thick] (6, 20) -- (34, 20);

\draw (10,55) node{$a_m$};
\draw (30,55) node{$b_m$};

\end{tikzpicture}\caption{Two cases of a new clique obtained from $G'$ by pivoting the edge $a_mb_m$ in Proposition~\ref{prop:noclique} where $r\in \{a_m, b_m\}$. By induction hypothesis, the each coloured part can have a clique of size at most $3^{d-2}$ in $G'\setminus r$, and therefore the size of a new clique cannot exceed $3^{d-1}$.}\label{fig:noclique}
\end{figure}
}
\begin{PROP}\label{prop:noclique}
	Let $d,t$ be positive integers such that $t> 3^{d-1}$.
	Then a graph of tree-depth at most $d$ cannot contain a pivot-minor isomorphic 
	to the clique $K_t$.  
\end{PROP}
\longversion{
\begin{proof}
	Let $K(d)=\max\,\{q:\td(G)\le d \text{ and $G$ has a pivot-minor isomorphic to $K_q$}\}$. 
	The statement is equivalent to $K(d)\le 3^{d-1}$.
	If $d=1$, then each component of a graph of tree-depth $1$ has one vertex and we have $K(1)=1$.
	We assume $d\ge 2$.

	We choose minimal $d$ such that a graph $G$ of tree-depth at most $d$ 
	has a pivot-minor isomorphic to $K_t$ where $t> 3^{d-1}$.
	Let $T$ be a tree-depth decomposition for $G$ of height at most $d$.
	Since $G$ is without loss of generality connected, $T$ has a unique
	root $r$ which is a vertex of $G$, too.
	Since $G$ has a pivot-minor isomorphic to $K_t$, 
	there exists $X\subseteq V(G)$ and $S\subseteq V(G)$ such that 
	\begin{enumerate}[a)]
	\item $\mx A(G)[X]$ is non-singular, and 
	\item the graph whose adjacency matrix is $(\mx A(G)*X)[S]$
	is isomorphic to $K_t$.
	\end{enumerate}
	By Lemma~\ref{lem:reorder}, for $s=r$ if $r\in X$ or $s\in X$ chosen
	arbitrarily otherwise, there exists a sequence of pairs of vertices
	$\{a_1, b_1\}, \{a_2, b_2\}, \ldots, \{a_m, b_m\}$ in $G$
	such that $\mx A(G)*X=\mx A(G\wedge a_1b_1\wedge a_2b_2 \cdots \wedge a_mb_m)$ and 
	$r\notin \{a_i,b_i\}$ for $1\le i\le m-1$.

	Let $G'=G\wedge a_1b_1\wedge a_2b_2 \cdots \wedge a_{m-1}b_{m-1}$.
	Then $(G'\wedge a_mb_m)[S]$ is isomorphic to $K_t$, and
	there are two cases:
	\begin{enumerate}[i.]
	\item
	$r\not\in\{a_m,b_m\}$, which means that $G\setminus r$ has the
	pivot-minor $(G'\wedge a_mb_m)\setminus r$ containing a
	$K_{t-1}$-subgraph.
	Since the tree-depth of $G\setminus r$ is $t-1$ as witnessed by
	the decomposition $T\setminus r$, and $t-1\geq3^{d-1}>3^{d-1-1}$, this
	contradicts our minimal choice of~$d$.
	\item
	$r=a_m$, up to symmetry.
	After the pivot $a_mb_m$, a new clique $K$ in $G$ (which is not present in $G'$) 
	is created in two possible ways:
	$K$ belongs to the closed neighbourhood of one of $a_m,b_m$,
	or $K$ is formed in the union of the neighbourhoods of $a_m,b_m$
	(excluding $a_m,b_m$).
	See Figure~\ref{fig:noclique}.
	In either case, $K$ is formed on two or three, respectively,
	cliques of $G'\setminus\{a_m,b_m\}$.
	Again, by minimality of $d$, the largest clique contained
        in $G'\setminus r$ can be of size ${3^{d-1-1}}$.
	Therefore, $t\leq\max\big(1+2\cdot3^{d-2},
		3\cdot3^{d-2}\big)=3^{d-1}$, a contradiction.
	\end{enumerate}
	Indeed, $t=K(d)\le 3^{d-1}$ as desired.
\end{proof}
}

\shortversion{
Hence, in contrast to rank-width~\cite{Kwon2013}, pivot-minors do not suffice to obtain all graphs of bounded
shrub-depth from graphs of bounded tree-depth. This raises the
question whether there is some simple condition on the graph class in
question which would guarantee us the theorem to hold for pivot-minors. It
turns out that one such simple restriction is to consider just bipartite
graphs of bounded shrub-depth, as stated by Theorem~\ref{thm:tdtoshrubd-bip}.


\begin{THM}
  \label{thm:tdtoshrubd-bip}
  For any class $\cf S$ of bounded shrub-depth, there exists an integer $d$
  such that every {\em bipartite} graph in $\cf S$ 
  is a pivot-minor of a graph of tree-depth~$d$.
\end{THM}
}

\longversion{

\section{Bipartite Graphs of small BSC-depth}

In the previous section we have seen that every graph class of bounded
shrub-depth can be obtained via vertex-minors of graphs of tree-depth $d$
for some $d$. Moreover, we have also proved that this statement does not
hold if we replace vertex-minors with pivot-minors. However this raises a
question whether there is some simple condition on the graph class in
question which would guarantee us the theorem to hold for pivot-minors. It
turns out that one such simple restriction is to consider just bipartite
graphs of bounded shrub-depth, as stated by Theorem~\ref{thm:tdtoshrubd-bip}.

To get our result, we introduce the following ``depth'' definition 
better suited to the pivot-minor operation, which builds upon the 
idea of SC-depth.
Let $G$ be a graph and let $X,Y\subseteq V(G)$, $X\cap Y=\emptyset$.
We denote by $\overline{G}^{(X,Y)}$ the graph $G'$ with vertex set $V(G)$ 
and edge set $E(G')=E(G)\Delta \{xy:x\in X, y\in Y\}$.
In other words, $\overline{G}^{(X,Y)}$ is the graph obtained from $G$ 
by complementing the edges between~$X$ and~$Y$.

\begin{DEF}[BSC-depth]
\label{def:BSC-depth}
We define inductively the class $\BSC k$ as follows:
\begin{enumerate}[i.]
  \item let $\BSC0=\{K_1\}$;
  \item if $G_1,\dots,G_p\in\BSC k$ and $H= G_1\dot\cup\dots\dot\cup G_p$,
  then for every pair of disjoint subsets $X,Y\subseteq V(H)$ we
  have $\overline{H\,}^{(X,Y)}\in\BSC{k+1}$.
\end{enumerate}
The {\em BSC-depth} of $G$ is the minimum integer $k$ such that $G\in\BSC k$.
\end{DEF}

\longversion{
In general, graphs of bounded SC-depth may have arbitrarily large
BSC-depth, but the two notions are anyway closely related,
as in Lemma~\ref{lem:BSC-SC-relation}.
Here $\chi(G)$ denotes the chromatic number of a graph.
\begin{LEM}\label{lem:BSC-SC-relation}
	\begin{enumerate}[\quad a)]
	\item The BSC-depth of any graph $G$ is at least $\lceil\log_2\chi(G)\rceil$.
	\label{it:logchi}
	\item The SC-depth of $G$ is not larger than three times its BSC-depth.
	\item If $G$ is bipartite, then the BSC-depth of $G$ is not larger 
		than its SC-depth.
	\label{it:bipSC}
	\end{enumerate}
\end{LEM}
\begin{proof}
a) If $H'=\overline{H\,}^{(X,Y)}$, then
$\chi(H')\leq2\chi(H)$ since one may use a fresh set of colours for the
vertices in~$Y$.
Then the claim follows by induction from Definition~\ref{def:BSC-depth}.

b) We have
$$
\overline{H\,}^{(X,Y)} = \overline{\left(
	\overline{\left(\overline{H\,}^X\right)}^Y
\right)}^{X\cup Y}
$$
and so the claim directly follows by comparing
Definitions~\ref{def:BSC-depth} and ~\ref{def:SC-depth}.

c) Let $G\in \SC k$.
Let $V(G)=A\cup B$ be a bipartition of $G$, i.e.,
that $A$ and $B$ are disjoint independent sets.
We use here for $G$ the same ``decomposition'' as in
Definition~\ref{def:SC-depth}; just replacing at every step
a single set $X$ with the pair $(X\cap A,X\cap B)$ (point ii. of the
definitions).
The resulting graph $G'\in\BSC k$ then fulfils the following:
both $A,B$ are independent sets in $G'$,
and every $uv\in A\times B$ is an edge in $G'$ if and only if
$uv$ is an edge of $G$.
Therefore, $G=G'\in\BSC k$.
\end{proof}

In particular, following Lemma~\ref{lem:BSC-SC-relation}\,\ref{it:logchi}),
the BSC-depth of the clique $K_n$ equals $\lceil\log_2n\rceil$,
while $K_{m,n}$ always have BSC-depth $1$.

\begin{LEM}\label{lem:BSC-from-pivot}
  Let $k$ be a positive integer.
  If a graph $G$ is of BSC-depth at most $k$, then
  $G$ is a pivot-minor of a graph of tree-depth at most $2k+1$.
\end{LEM}
\begin{proof}
  The proof follows along the same line as the proof of
  Lemma~\ref{lem:vertex-minor-tree-depth}.
  For a graph $G$ of BSC-depth $k$,
  we recursively construct a graph $U$ and a rooted forest $T$ such that
  \begin{enumerate}[i.]
  \item $G$ can be obtained from $U$ as a pivot-minor via pivoting edges
    only between vertices in $V(U)\setminus V(G)$, and
  \item $U\subseteq \cl(T)$ and $T$ has depth at most $2k+1$.
  \end{enumerate}
  If $k=0$, then it is clear by setting $G=U=T=K_1$. 
  We assume that $k\ge 1$.
  
  Since $G$ has BSC-depth $k$, there exist a graph $H$ and disjoint
  subsets $X,Y\subseteq V(H)$ such that $G=\overline{H}^{(X,Y)}$
  and $H$ is the disjoint union of $H_1, H_2, \ldots, H_m$ such that each $H_i$ has BSC-depth $k-1$.
  By induction hypothesis, for each $1\le i\le m$, $H_i$ is a pivot-minor of a graph $U_i$ and $U_i\in \cl(T_i)$ where the
  height of $T_i$ is at most $2(k-1)+1$.
  For each $1\le i\le m$, let $r_i$ be the root of $T_i$, and let $T$
  be the rooted forest obtained from the disjoint union of all $T_i$
  by adding an edge between two new vertices $r_x$ and $r_y$ and by
  connecting $r_Y$ to all $r_i$.
  Let $U$ be the graph obtained from the disjoint union of all $U_i$
  and the vertices $\{r_x,r_y\}$ by adding an edge between $r_x$ and $r_y$ and
  all edges from $r_x$ to $X$ as well as all edges from $r_y$ to $Y$. 
  Validity of (ii.) is clear from the construction.
  
  Now we check the statement (i.). 
  By our construction of $U$, any pivoting on edges in $U_i$ has no effect
  on $U_j$ for $j\not=i$, and pivoting on edges
  in $V(U_i)\setminus V(H_i)$ does not change edges incident with
  $r_x$ or $r_y$.
  Hence, by induction, we can obtain $H$ as a pivot-minor of $U$ and
  still have $r_x$ adjacent precisely to $r_y$ and $X\subseteq V(H)$ and $r_y$ 
  adjacent to $r_x$ and $Y \subseteq V(H)$.
  We then pivot the edge $\{r_x,r_y\}\in V(U)\setminus
  V(H)$, and delete $V(U)\setminus V(G)$ to obtain $G$.
\end{proof}
}

\shortversion{Because for bipartite graphs the BSC-depth is bounded by
the SC-depth of a graph, we obtain the main result of this section
immediately from an equivalent of Lemma~\ref{lem:vertex-minor-tree-depth} for
BSC-depth and Proposition~\ref{thm:shrubsc}.}

\longversion{The main result of this section now immediately follows from
Lemmas~\ref{lem:BSC-from-pivot}, \ref{lem:BSC-SC-relation}\,\ref{it:bipSC})
and Proposition~\ref{thm:shrubsc}.}

\begin{THM}
\label{thm:tdtoshrubd-bip}
For any class $\cf S$ of bounded shrub-depth, there exists an integer $d$
such that every {\em bipartite} graph in $\cf S$ 
is a pivot-minor of a graph of tree-depth~$d$.
\end{THM}
\longversion{\smallskip}
}

\section{Conclusions}

We finish the paper with two questions that naturally arise from our
investigations.
While the first question has a short negative answer, the second one is left
as an open problem.

A {\em cograph} is a graph obtained from singleton vertices by repeated
operations of disjoint union and (full) complementation.
This well-studied concept has been extended to so called ``$m$-partite
cographs'' in \cite{GHNOOR12} (we skip the technical definition here for
simplicity); where cographs are obtained for $m=1$.
It has been shown in \cite{GHNOOR12} that $m$-partite   
cographs present an intermediate step between classes of bounded shrub-depth
and those of bounded clique-width.

The first question is whether some of our results can be extended from classes of
bounded shrub-depth to those of $m$-partite cographs.
We know that shrub-depth is monotone under taking vertex-minors
(Corollary~\ref{cor:shrubd-monotone}) and an analogous claim is
asymptotically true also for clique-width \cite{os06}.
However, the main obstacle to such an extension is the fact that $m$-partite
cographs do not behave well with respect to local and pivot equivalence of
graphs. To show this we will employ the following proposition:

\begin{PROP}[\cite{GHNOOR12}]
A path of length $n$ is an $m$-partite cograph if, and only if, $n<3(2^m-1)$.
\end{PROP}

By the proposition, to negatively answer our question it is enough to find a
class of $m$-partite cographs containing long paths as pivot-minors:

\shortversion{
\begin{PROP}
\label{prop:Hn-path}
For every $n\geq1$, the graph $H_n$ from~\cite[Example 5.4 a)]{GHNOOR12} is a cograph that
contains a path of length $n$ as a pivot-minor.
\end{PROP}
}

\longversion{
\begin{PROP}
\label{prop:Hn-path}
Let $H_n$ denote the graph on $2n$ vertices from Figure~\ref{fig:cograph2n}.
Then $H_n$ is a cograph for each $n\geq1$,
and $H_n$ contains a path of length $n$ as a pivot-minor.
\end{PROP}

\begin{figure}[t]
\begin{center}
\includegraphics[width=.55\textwidth]{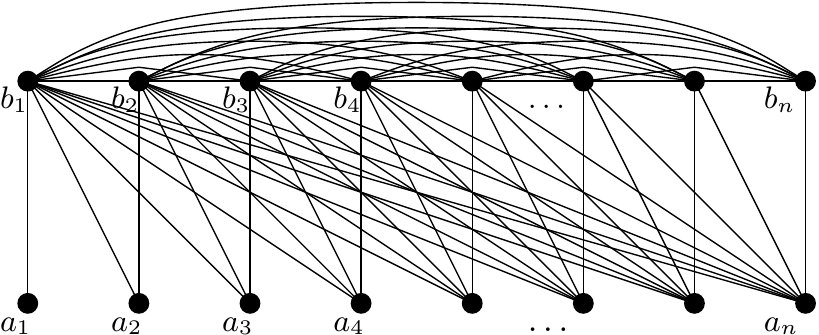}
\end{center}
\caption{A graph on $2n$ vertices \cite{GHNOOR12} which is a cograph
and pivoting on $a_2b_2,a_3b_3,\dots,a_{n-1}b_{n-1}$ results in an induced path
on $a_1,b_1,b_2,\dots,b_n$.}
\label{fig:cograph2n}
\end{figure}




\begin{proof}
It is $V(H_n)=\{a_i,b_i: i=1,2,\dots,n\}$ and
$E(H_n)=\{b_ib_j:$ $1\leq i<j\leq n\}\cup
 \{b_ia_j:$ $1\leq i\leq j\leq n\}$.
The graph $H_n$ can be constructed iteratively as follows, for $j=n,n-1,\dots,1$:
add a new vertex $a_j$, complement all the edges of the graph,
add a new vertex $b_j$, complement again.
Consequently, $H_n$ is a cograph (and, in fact, a so called threshold graph).

For the second part, we let inductively $G_1:=H_n$
and $G_{j}:=G_{j-1}\wedge a_jb_j$ for $j=2,\dots,n-1$.
Then, by the definition, $G_{2}$ is obtained from $H_n$ by removing
all the edges incident with $b_1$ except $b_1a_2,b_1b_2$.
In particular, $G_{2}\setminus\{a_1,b_1\}$ is isomorphic to $H_{n-1}$,
and $a_3,b_3$ are adjacent in $G_{2}$ only to vertices other than $a_1,b_1$.
Consequently, by induction, $G_j$ is obtained from $G_{j-1}$ by removing all
the edges incident with $b_{j-1}$ except $b_{j-1}a_{j},b_{j-1}b_{j}$,
and $G_{n-1}$ has the edge set $\{b_1b_2,b_2b_3,\dots,b_{n-1}b_n\}\cup
 \{a_1b_1,a_2b_1,a_2b_2,a_3b_2,\dots,a_nb_n\}$.
Then $G_{n-1}[a_1,b_1,b_2,\dots,b_n]$ is a path.
\end{proof}
}

Building on this negative result, it is only natural to ask whether not having 
a long path as vertex-minor is the property exactly characterising shrub-depth.

\begin{CONJ}
     A class $\cf C$ of graphs is of bounded shrub-depth if, and only if,
     there exists an integer $t$ such that no graph $G\in\cf C$
     contains a path of length $t$ as a vertex-minor. 
\end{CONJ}



\bibliographystyle{abbrv}
\bibliography{gtbib}

\end{document}